\documentclass[smallextended]{svjour3}       

\smartqed  

\usepackage{amssymb,amsmath}
\usepackage{graphicx}
\usepackage{booktabs}
\usepackage{hyperref}
\usepackage{enumerate, multirow}

\begin{document}

\title{Partially user-irrepressible sequence sets and conflict-avoiding codes}


\author{Yuan-Hsun Lo \and Wing Shing Wong \and Hung-Lin Fu}


\institute{Y.-H. Lo \at
              Department of Mathematics, National Taiwan Normal University, Taipei 116, Taiwan \\
              \email{yhlo0830@gmail.com}
           \and
           W. S. Wong \at
           	  Department of Information Engineering, the Chinese University of Hong Kong, Hong Kong\\
           	  \email{wswong@ie.cuhk.edu.hk}
           \and
           H.-L. Fu \at
              Department of Applied Mathematics, National Chiao Tung University, Hsinchu 300, Taiwan\\
              \email{hlfu@math.nctu.edu.tw}
           \and
              The work partially supported by Research Grants Council of the Hong Kong Special Administrative Region under project 414012 (Y.-H. Lo and W. S. Wong), and the National Science Council under grants 100-2115-M-009-005-MY3 (H.-L. Fu).
}

\date{Received: date / Accepted: date}

\maketitle

\begin{abstract}
In this paper we give a partial shift version of user-irrepressible sequence sets and conflict-avoiding codes.
By means of disjoint difference sets, we obtain an infinite number of such user-irrepressible sequence sets whose lengths are shorter than known results in general.
Subsequently, the newly defined partially conflict-avoiding codes are discussed.
\keywords{user-irrepressible protocol sequence \and conflict-avoiding code \and disjoint difference set}
\subclass{94B25 \and 94C15 \and 05B10}
\end{abstract}

\section{Introduction}\label{sec:Intro}
Protocol sequences, which were first introduced in~\cite{Massey}, provide feedback-free solutions for Media Access Control (MAC) in communication networks.
While the dominant MAC standards for cell-based systems, including cellular networks and Wireless LAN's, are feedback-based, the feedback-free approach has a strong appeal to networks without a backbone hierarchy. 
For example, recent works have begun to explore the application of protocol sequences to ad hoc networks, such as \emph{vehicular ad hoc network} (VANET)~\cite{Wong_2014,Wu_Shum_Wong_Shen_2014}.

A fundamental challenge in MAC design is due to the lack of synchronicity among different users who try to access the shared medium.
Protocol sequences are constructed specifically to handle the asynchronous reality.
Intuitively, a good design should ensure that no matter how the sequences are shifted with respect to one another, each sequence should permit its affiliated user to transmit at least one packet without suffering interference from other users.
Protocol sequence sets with this property are commonly referred to as possessing the user-irrepressible (UI) property~\cite{Shum_Wong_2009,Wong_2007}.
It turns out that an important approach to construct UI protocol sequence sets is by means of CAC, which stands for Conflict-avoiding Codes \cite{Fu_Lo_Shum_2014,Momihara_2007,Shum_Wong_Chen_2010}.
Therefore, there is a close tie between protocol sequences and CAC.
The objective of finding UI protocol sequence sets with large number of sequence elements with short sequence period can be transformed to finding CAC sets with large code size and short code length.

Although it is difficult to ensure precise user-synchronicity in multi-user communication systems, in many applications it is relatively easy to maintain some rough degree of user synchronicity.
For example, mobile users may have access to a global clock via the GPS, which provides rough time synchronization.
However, due to propagation delays and other engineering restrictions, transmitted signals cannot be completely synchronized (see for example \cite{Wong_2014}).
For partially synchronous applications, protocol sequence sets are only required to observe the UI property for relative shifts up to a certain magnitude.

In this paper, we define a partial shift version of user-irrepressible sequence sets in Section~\ref{sec:UI}.
Two prior known constructions: TDMA and code-based scheduling (via Galois field or Reed-Solomon code), are then
introduced to provide some quick baseline comparison.
Next, we introduce a new concept, called partially conflict-avoiding code (PCAC), in order to build a partially user-irrepressible sequence set.
The definition of a partially conflict-avoiding code will be given in Section~\ref{sec:combin} together with its graphic representation.
A useful tool in combinatorial design called disjoint difference set is also introduced.
In Section~\ref{sec:MainResult} we provide a few families of partially user-irrepressible sequence sets by means of disjoint difference sets.
Comparison of the PCAC approach with TDMA and code-based scheduling will also be given in Section~\ref{sec:MainResult}.
Finally, we study the optimal partially conflict-avoiding codes of small weights in Section~\ref{sec:PCAC23}.

\section{User-Irrepressible sequences}\label{sec:UI}

Let $n$ be a positive integer and $X$ be a binary sequence of length $n$.
The \emph{cyclic shift operator}, $\mathcal{R}$, on $X$ is defined by
$$\mathcal{R}(X(0),X(1),\ldots,X(n-1)):=(X(n-1),X(0),\ldots,X(n-2)),$$
where $X(i)$ denotes the $i$-th component of $X$.
The following definition is an extension of \emph{user-irrepressible} property which is proposed in \cite{Shum_Wong_Chen_2010}.

\begin{definition}
Let $n,k,\Delta$ be integers satisfying $0<k\leq n$ and $0\leq\Delta<n$.
Consider a sequence set with $N$ ($\geq k$) elements, each having a length $n$.
Each element is represented by a shifted version that is obtained by applying the operator $\mathcal{R}$ for an arbitrary number (say $\tau$) of times, where $0\leq\tau\leq\Delta$.
Denote by $\mathbf{M}$ the $k\times n$ matrix obtained by stacking any $k$ representations one above the other.
The sequence set is \emph{$(n,k;\Delta)$-User-Irrepressible} (UI for short) if we can always find a $k\times k$ submatrix of $\mathbf{M}$ which is a permutation matrix.
\end{definition}

An $(n,k;\Delta)$-UI sequence set is obviously a solution to the problem we formulated in Section~\ref{sec:Intro}.
Throughout this paper, we use $N$, $k$, and $n$ to denote respectively the number of potential users in a system, the maximum number of active users at any time, and the common sequence period.

It is not hard to find an $(n,k;\Delta)$-UI sequence set.
One simple way is based on the TDMA approach.
For $0\leq i \leq k-1$, let $X_i$ be the binary sequence of length $k(\Delta+1)$ composed of all zeroes except for the $i(\Delta+1)$-th position, that is, $X_i(i(\Delta+1))=1$.
Then $\{X_0,X_1,\ldots,X_{k-1}\}$ is obviously an $(n,k;\Delta)$-UI sequence set of length $n=k(\Delta+1)$ and size $N=k$.
In practice, however, the set size $N$ is in theory larger than $k$.
An alternative construction for the case where $N$ is much larger than $k$ is based on Galois fields.
After appending $\Delta$ `zeroes' to all entries of each sequence constructed in \cite{Chlamtac_1994}, we have the following result.

\begin{theorem}(\cite{Chlamtac_1994}, \cite{Wong_2014}) \label{thm:Galois_field}
Given a prime power $q$ and a positive integer $m$.
Then for any $\Delta\geq 0$, there exists a $((\Delta+1)q^2,k;\Delta)$-UI sequence set of size $N=q^m$, where the positive integer $k$ satisfies
\begin{equation}\label{eq:GF_condition}
q\geq (k-1)(m-1)+1.
\end{equation}
In general, it provides an $(n,k;\Delta)$-UI sequence set of size $N$ with length
$$n = O\left( \Delta k^2m^2 \right) = O\left( \frac{\Delta k^2\ln^2N}{\ln^2k} \right).$$
\end{theorem}

Note that the parameter $m$ above must be larger than $1$ to make \eqref{eq:GF_condition} meaningful.
It is worth mentioning that in \cite{Rentel_Kunz_2005}, a solution based on Reed-Solomon Codes was proposed which has the same order behavior.

\section{Combinatorial structure}\label{sec:combin}

In this section, we define the new concept of partially conflict-avoiding codes and introduce
two relevant combinatorial structures for analyzing them: graph packings and disjoint difference sets.
The connection of these terms with UI sequence sets will be shown as

\begin{align}
\begin{array}{rcc}
(n,k;\Delta)\text{-UI sequence set} & {\Longleftarrow} \atop \text{Prop.~\ref{pro:UI_PCAC}} & \text{PCAC}_{\Delta}(n,k) \\
 & & \hspace{0.8cm}\Updownarrow \substack{\text{Prop.~\ref{pro:PCAC_packing}}} \\
(n,k,r)\text{-DDS} & {\Longrightarrow} \atop \text{Prop.~\ref{pro:DDS_packing}} & (k,\Delta)\text{-packing} \text{ of }K_n
\end{array}
\end{align}

\medskip

\subsection{CAC and $\text{PCAC}_\Delta$}

Given a binary sequence $X$, the \emph{weight} of $X$, denoted by $\omega(X)$, is the number of `ones' in it.
For integers $n>k>0$, let $\mathcal{S}(n,k)$ denote the set of all binary sequences of length $n$ and weight $k$.
The \emph{Hamming cross-correlation} of binary sequences $X$ and $Y$ is defined by
\begin{equation} \label{eq:full-crosscorel}
H(X,Y):= \max_{\tau}\sum_{i=0}^{n-1}X(i)\mathcal{R}^\tau Y(i),
\end{equation}
where $\tau$ goes from $0$ up to $n-1$.
Note that $H(X,X)=\omega(X)$ for all $X$ and $H(X,Y)\geq 1$ if $X\neq Y$.

\begin{definition} \label{defi:CAC}
A set $\mathcal{C}\subseteq\mathcal{S}(n,k)$ is a \emph{conflict-avoiding code}, CAC, of length $n$ and weight $k$ if $H(X,Y)=1$ for any distinct $X,Y\in\mathcal{C}$.
\end{definition}

Denote by CAC$(n,k)$ the class of all CACs of length $n$ and weight $k$.
The maximum size of codes in CAC$(n,k)$ is denoted by $M(n,k)$.
A code $\mathcal{C}\in \mbox{CAC}(n,k)$ is said to be \emph{optimal} if $|\mathcal{C}|=M(n,k)$.
For more results on optimal CACs, please refer to \cite{Fu_Lin_Mishima_2010,Fu_Lo_Shum_2014,Levenshtein_Tonchev_2005,Momihara_2007,Shum_Wong_Chen_2010}.

In what follows, we generalize the constraint that $\tau$ is arbitrary in \eqref{eq:full-crosscorel}.
Assume that $\Delta$, an integer between 0 and $n-1$, is the maximum number of relative cyclic shifts.
Then the \emph{Hamming cross-correlation of $X,Y\in\mathcal{S}(n,k)$ with respect to $\Delta$} is defined by
\begin{equation}
H_{\Delta}(X,Y):= \max_{0\leq\tau \leq\Delta}\sum_{i=0}^{n-1}X(i)\mathcal{R}^\tau Y(i).
\end{equation}

\begin{definition} \label{defi:partial-CAC}
Let $n,k,\Delta$ be integers with $0<k<n$ and $0\leq\Delta<n$.
A set $\mathcal{C}\subseteq\mathcal{S}(n,k)$ is a \emph{partially conflict-avoiding code with respect to $\Delta$}, $\mbox{PCAC}_{\Delta}$, of length $n$ and weight $k$ if $H_{\Delta}(X,Y)\leq 1$ for any distinct $X,Y\in \mathcal{C}$.
\end{definition}

Similarly, $\text{PCAC}_{\Delta}(n,k)$ denotes the class of all $\mbox{PCAC}_{\Delta}$s of length $n$ and weight $k$, and $M_{\Delta}(n,k)$ denotes the maximum size of codes in $\mbox{PCAC}_{\Delta}(n,k)$.
It is obvious that a $\text{PCAC}_{\Delta}$ admits the UI-property.

\begin{proposition}\label{pro:UI_PCAC}
A code $\mathcal{C}\in\text{PCAC}_{\Delta}(n,k)$ is an $(n,k;\Delta)$-UI sequence set with size $N=|\mathcal{C}|$.
\end{proposition}

Let $n,k,\Delta$ be integers satisfying the setting of Definition~\ref{defi:partial-CAC}.
It is clear that 
$$\text{PCAC}_{\Delta}(n,k)\supseteq \text{PCAC}_{\Delta+1}(n,k)\supseteq\cdots\supseteq \text{PCAC}_{n-1}(n,k) = \text{CAC}(n,k),$$
and thus 
$$M_{\Delta}(n,k)\geq M_{\Delta+1}(n,k)\geq\cdots\geq M_{n-1}(n,k) = M(n,k).$$
Here is an interesting observation.

\begin{lemma} \label{lem:large-delta}
Let $n,k$ be integers with $n>k>0$.
If $\Delta$ is an integer with $\lfloor\frac n2\rfloor \leq \Delta < n$, then $M_{\Delta}(n,k)=M(n,k)$.
\end{lemma}
\begin{proof}
We first claim that $H_\Delta(X,Y)\geq 1$ for any two distinct sequences $X,Y$ in $\mathcal{S}(n,k)$.
Assume to the contrary that $H_\Delta(X,Y)=0$. Pick any two indices $i,j$ with $X(i)=Y(j)=1$.
For every $\tau=0,1,\ldots,\Delta$, since $X(i)\mathcal{R}^\tau Y(i)=0$, we have $Y(i-\tau)=0$, where the addition is taking modulo $n$.
Similarly, there are consecutive $\Delta+1$ `zeroes' from $X(j-\Delta)$ to $X(j)$.
Since $X(i)=Y(j)=1$, those $2(\Delta+1)$ indices are distinct (see Figure~\ref{fig:large-delta}).
Then we have $2(\Delta+1)\leq n$, which contradicts to $\lfloor\frac n2\rfloor \leq \Delta$.

\begin{figure}[h]
\centering
\includegraphics[width=3.5in]{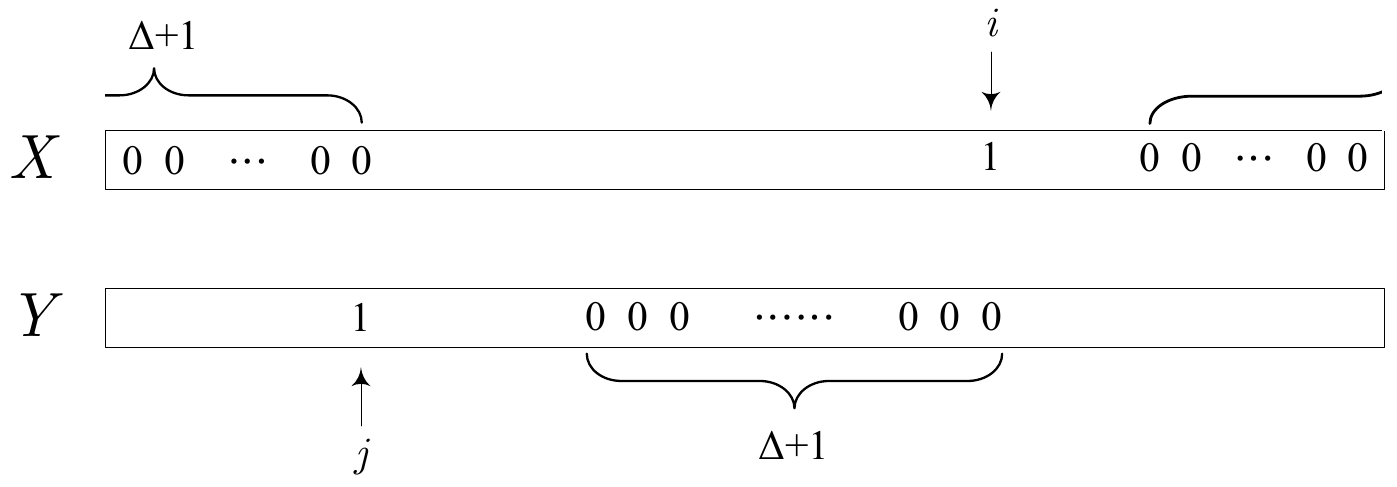}
\caption{Illustration of $X(i)=Y(j)=1$} \label{fig:large-delta}
\end{figure}

Let $\mathcal{C}\in \mbox{PCAC}_{\Delta}(n,k)$.
Above argument promises that $H_\Delta(X,Y)=1$ for any two distinct sequences $X,Y\in\mathcal{C}$.
We now claim that $\mathcal{C}\in \mbox{CAC}(n,k)$.
Assume to the contrary that there exist two distinct sequences $X,Y\in\mathcal{C}$ so that $H(X,Y)\geq 2$.
By symmetry there exist indices $i_1,i_2,j_1,j_2$ such that $X(i_1)=X(i_2)=1$ and $Y(j_1)=Y(j_2)=1$, where $i_1+\tau\equiv j_1$ (mod $n$) and $i_2+\tau\equiv j_2$ (mod $n$) for some $\tau\leq \Delta$. This contradicts to $H_{\Delta}(X,Y)= 1$.
Hence the proof is completed.
\qed
\end{proof}

\subsection{Graphic representation}\label{sec:PCAC_graph}

Let $\mathbb{Z}_n=\{0,1,\ldots,n-1\}$ denote the ring of residues modulo $n$.
Let $K_n$ denote the complete graph of order $n$ whose vertices are labeled by elements in $\mathbb{Z}_n$.
Given any subset $A\subseteq\mathbb{Z}_n$, let $C_A$ denote the \emph{clique} induced by $A$, namely, the subgraph with vertex set $A$ whose vertices are pairwise adjacent.
A clique of order $t$ is usually called a $t$-clique.
Given an integer $\Delta$ with $0\leq\Delta < n$, the \emph{supporting graph} of $A$ with respect to $\Delta$ is defined as
$$G_{\Delta}(A):= C_{A}\cup C_{A+1}\cup \cdots \cup C_{A+\Delta},$$
where $A+\tau=\{i+\tau ~(\mbox{mod }n) :\, i\in A\}$.
By putting the $n$ vertices of $K_n$ in clockwise direction from $0$ to $n-1$, $G_{\Delta}(A)$ can be viewed as the union of $(\Delta+1)$ $|A|$-cliques, each of which is obtained by rotating $C_{A}$ clockwise step by step.
For example, let $n=8$, $\Delta=2$ and $A=\{0,1,2\}$, $B=\{3,5,7\}$, then $A+1=\{1,2,3\}$, $A+2=\{2,3,4\}$, $B+1=\{4,6,0\}$ and $B+2=\{5,7,1\}$.
See Figure~\ref{fig:supporting} for the two supporting graphs: $G_{2}(A)$ and $G_{2}(B)$.

\begin{figure}[h]
\centering
\includegraphics[width=3.5in]{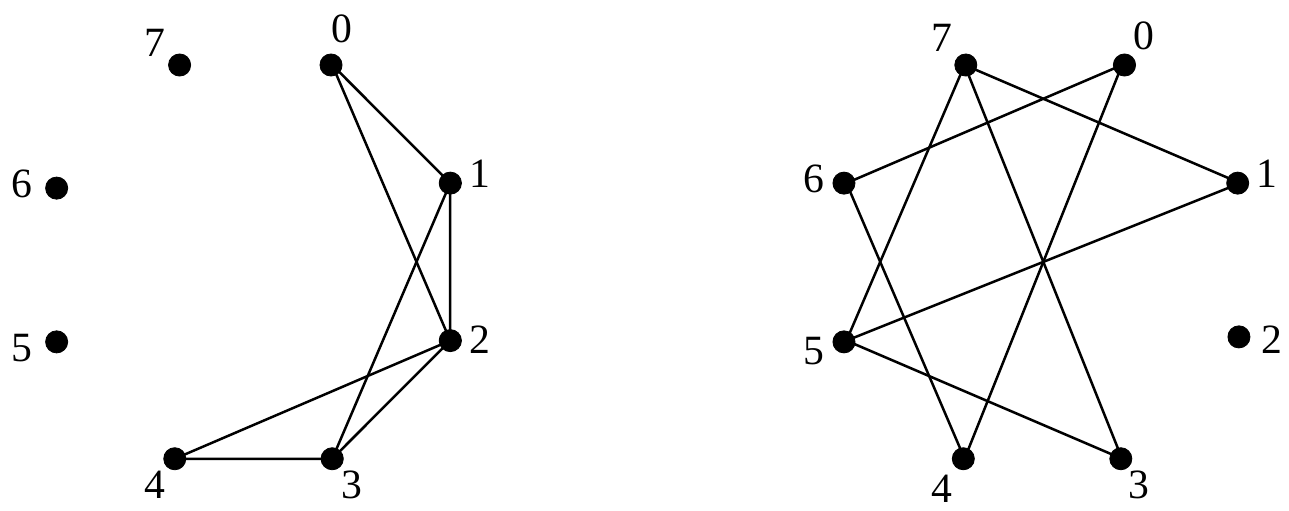}
\caption{$G_{2}(\{0,1,2\})$ and $G_{2}(\{3,5,7\})$ in $K_8$} \label{fig:supporting}
\end{figure}

For a binary sequence $X$ of length $n$, the \emph{characteristic set} of $X$ is given by
$$\mathcal{I}_X :=\{t\in\mathbb{Z}_n :\, X(t)=1\}.$$
A cyclic shift of $X$ by $\tau$ corresponds to a translation of $\mathcal{I}_X$ by $\tau$ in $\mathbb{Z}_n$, that is, $\mathcal{I}_{\mathcal{R}^\tau X} = \mathcal{I}_X + \tau$.
Let $n,k,\Delta$ be integers with $0<k<n$ and $0\leq\Delta<n$.
Given two distinct binary sequences $X,Y\in\mathcal{S}(n,k)$, it is easy to see that $H_{\Delta}(X,Y)\leq 1$ if and only if $G_{\Delta}(\mathcal{I}_X)$ and $G_{\Delta}(\mathcal{I}_Y)$ are edge-disjoint.

\begin{definition}\label{defi:packing}
Let $\mathcal{P}=\{P_1,P_2,\ldots,P_N\}$ be a set of $k$-subsets of $\mathbb{Z}_n$.
We say $\mathcal{P}$ is a \emph{$(k,\Delta)$-packing} of $K_n$ if $G_{\Delta}(P_i)$ and $G_{\Delta}(P_j)$ are edge-disjoint whenever $i\neq j$.
\end{definition}

The following follows directly from definitions.

\begin{proposition}\label{pro:PCAC_packing}
Let $n,k,\Delta$ be integers with $0<k<n$ and $0\leq\Delta<n$.
There exists a code $\mathcal{C}\in\text{PCAC}_{\Delta}(n,k)$ with $|\mathcal{C}|=N$ if and only if $K_n$ has a $(k,\Delta)$-packing $\mathcal{P}=\{P_1,P_2,\ldots,P_N\}$.
More precisely, $\mathcal{P}=\{\mathcal{I}_X:\,X\in\mathcal{C}\}$.
\end{proposition}

A $(k,\Delta)$-packing $\mathcal{P}$ of $K_n$ is said to be \emph{maximum} if the size of $\mathcal{P}$ is maximum.
That is, a maximum $(k,\Delta)$-packing of $K_n$ is equivalent to an optimal $\text{PCAC}_\Delta$ of length $n$ and weight $k$.

\subsection{Disjoint difference set}

\begin{definition}\label{defi:DDS}
An \emph{$(n,k,r)$-disjoint difference set} (\emph{DDS}) is a family $\{B_1,B_2,$ $\ldots,B_r\}$ of $k$-subsets of $\mathbb{Z}_n$ such that among the differences $\{x-y:\,x,y\in B_i,x\neq y,1\leq i\leq r\}$ each nonzero element $g\in\mathbb{Z}_n$ occurs at most once.
\end{definition}

A necessary condition for the existence of an $(n,k,r)$-DDS is
\begin{equation}\label{eq:DDS_necessary}
n\geq r k(k-1)+1.
\end{equation}
An $(n,k,r)$-DDS is called as an \emph{$(n,k)$-difference family} (DF) if the equality in \eqref{eq:DDS_necessary} holds.
That is, an $(n,k)$-DF is an $(n,k,\frac{n-1}{k(k-1)})$-DDS.

Let $\{B_1,B_2,\ldots,B_r\}$ be an $(n,k,r)$-DDS.
It is easy to check that for any $\Delta,t,t'\geq 0$, the two cliques $C_{B_i+t}$ and $C_{B_i+t'}$ have no common edges whenever $t\neq t'$, and the two supporting graphs $G_\Delta(B_i+t)$ and $G_\Delta(B_j+t')$ are edge-disjoint whenever $i\neq j$.
Hence, we have the following proposition.

\begin{proposition}\label{pro:DDS_packing}
Let $\{B_1,B_2,\ldots,B_r\}$ be an $(n,k,r)$-DDS.
For $0\leq\Delta<n$, there exists a $(k,\Delta)$-packing of $K_n$ with size $r\left\lfloor\frac{n}{\Delta+1}\right\rfloor$.
\end{proposition}
\begin{proof}
By the observation above, the set of supporting graphs $G_\Delta(B_i+t)$ for $i=1,2,\ldots,r$ and $t=0,(\Delta+1),2(\Delta+1),\ldots, (\lfloor\frac{n}{\Delta+1}\rfloor-1)(\Delta+1)$ will form a $(k,\Delta)$-packing of $K_n$.
This concludes the proof.
\qed
\end{proof}

Combining Proposition~\ref{pro:UI_PCAC}, \ref{pro:PCAC_packing} and \ref{pro:DDS_packing}, we conclude that

\begin{theorem}\label{thm:UI_DF}
If there exists an $(n,k,r)$-DDS, then for $0\leq\Delta<n$, there exists an $(n,k;\Delta)$-UI sequence set of size 
\begin{equation}\label{eq:UI_size}
N=r\left\lfloor\frac{n}{\Delta+1}\right\rfloor.
\end{equation}
\end{theorem}

\smallskip
In order to obtain $(n,k,r)$-DDSs, we revisit a useful combinatorial structure called difference triangle sets.

%
%

\begin{definition} \label{defi:DTS}
A \emph{normalized $(r,k)$-difference triangle set} (\emph{DTS} for short) is a family $\{B_1,B_2,\ldots,B_r\}$, where $B_i=\{b_{i0},b_{i1},\ldots,b_{ik}\}$, $1\leq i\leq r$, are sets of integers such that $0=b_{i0}<b_{i1}<\cdots<b_{ik}$, for all $i$, and such that the differences $b_{ij'}-b_{ij}$ with $1\leq i\leq r$ and $0\leq j<j'\leq k$ are all distinct.
The \emph{scope} of an $(r,k)$-DTS is the maximum integer among $\{b_{1k},b_{2k},\ldots,b_{r k}\}$.
\end{definition}

It is known that a DDS can be obtained from a DTS.

\begin{theorem}\cite{Shearer_2007}\label{thm:DDS_DTS}
An $(r,k-1)$-DTS of scope $m$ is an $(n,k,r)$-DDS for all $n\geq 2m+1$.
\end{theorem}

Please refer to \cite{Chen_1994,Chee_Colbourn_1997,Chu_Colbourn_Golomb_2005,Chen_Fan_Jin_1992,Ling_2002,Shearer_2007} for more information on DDSs and DTSs.
Note that a DDS is also named as a \emph{difference packing (DP)} in literature.

\subsection{An example}
We use an example to illustrate our idea.
Suppose that we aim to construct a $(19,3;5)\text{-UI}$ set of size as large as possible.
The first step is to find a $(19,3,3)\text{-DDS}: B_1=\{0,4,5\},B_2=\{0,6,8\},B_3=\{0,7,10\}$.
Note that $\{B_1,B_2,B_3\}$ forms a difference family.
By Proposition~\ref{pro:DDS_packing}, we have a $(3,5)$-packing of $K_{19}$ as follows:
\begin{align*}
\text{From } B_1&: \{0,4,5\}, \{6,10,11\}, \{12,16,17\},\\
\text{From } B_2&: \{0,6,8\}, \{6,12,14\}, \{12,18,1\},\\
\text{From } B_3&: \{0,7,10\}, \{6,13,16\}, \{12,0,3\}.
\end{align*}
Therefore, by Proposition~\ref{pro:UI_PCAC} and \ref{pro:PCAC_packing}, the $9$ desired sequences are listed below.
$$
\begin{array}{crrr}
B_1: & 1000110000000000000, & 0000001000110000000, & 0000000000001000110, \\
B_2: & 1000001010000000000, & 0000001000001010000, & 0100000000001000001, \\
B_3: & 1000000100100000000, & 0000001000000100100, & 1001000000001000000.
\end{array}
$$

Let us consider a network of $9$ potential users with the constraint that at most $3$ of them are active at the same time and the maximum relative shift is $5$.
Then above example (PCAC approach) provides a solution with sequence length $n=19$.
If we consider TDMA approach, the length of sequences must be larger than $9\times 5=45$.
If we consider GF (or RS code) approach, by taking $k=3, \Delta=5$ and $N\geq 9$ into Theorem~\ref{thm:Galois_field}, we have $m\geq 2$ and $q\geq 3$, and thus $n\geq (5+1)\times 3^2=54$.
This indicates that applying PCAC approach is more efficient than the other two methods.
We will study this phenomenon in more details in the subsequent section.

\subsection{Remarks}
It must be noted that the connection in Proposition~\ref{pro:DDS_packing} is an old fashion.
In fact, such a link is widely used to construct a block design from a difference family, see \cite{Chen_Fan_Jin_1992,Shearer_2007}. 
However, it is new to connect it with CAC or protocol sequences.
If we let $D(B)$ denote the set of differences of any two elements in a set $B\subset\mathbb{Z}_n$, then any two sequences $X$ and $Y$ in a CAC have the property that $D(\mathcal{I}_X)\cap D(\mathcal{I}_Y) = \emptyset$.
Since the quantity of sequences is what counts here, a good (or optimal) CAC is designed to make sure each $|D(\mathcal{I}_X)|$ is as small as possible, which is different from the demand of a difference family or a disjoint difference set.

\section{New construction of UI sequence sets}\label{sec:MainResult}

In this section, we first construct a few families of UI sequence sets by means of disjoint difference sets, and then compare them with the UI sequence sets produced in Section~\ref{sec:UI}.

Singer~\cite{Singer_1938} constructed $(q^2+q+1,q+1,1)$-DDS, and Bose~\cite{Bose_1942} constructed $(q^2-1,q,1)$-DDS, where $q$ is a prime power. With these DDSs and a construction of Colbourn-Bolbourn~\cite{Colbourn_1984}, Chen-Fan-Jin~\cite{Chen_Fan_Jin_1992} proposed two infinite families of disjoint difference sets.

\begin{theorem}\cite{Chen_Fan_Jin_1992}
Let $q$ be a prime power.
\begin{enumerate}[(a)]
\item There exists an $\left(r(q^2+q+1),q+1,r\right)$-DDS for any prime $r>q$.
\item There exists an $\left(r(q^2-1),q,r\right)$-DDS for any prime $r\geq q$.
\end{enumerate}
\end{theorem}

%
%
%
%

By Theorem~\ref{thm:UI_DF}, we have the following result.

\begin{theorem} \label{thm:constant-k}
Let $q$ be a prime power.
\begin{enumerate}[(a)]
\item For $r=1$ or $r>q$ is a prime, there exists an $\left(r(q^2+q+1),q+1;\Delta\right)$-UI sequence set with size $$N=r\left\lfloor\frac{r(q^2+q+1)}{\Delta+1}\right\rfloor.$$
\item For $r=1$ or $r\geq q$ is a prime, there exists an $\left(r(q^2-1),q;\Delta\right)$-UI sequence set with size $$N=r\left\lfloor\frac{r(q^2-1)}{\Delta+1}\right\rfloor.$$
\end{enumerate}
\end{theorem}

Theorem~\ref{thm:constant-k} provides a new method to construct $(n,k;\Delta)$-UI sequence sets for some particular $n$.
We now investigate the properties of the three constructions: PCAC, TDMA and GF (or RS code) methods.
See the following chart for the comparisons.

\begin{table}[h]
\setlength{\belowcaptionskip}{0pt}
\scriptsize
\begin{tabular}{|c||c|c|c|c|}
\hline
 & potential users & sequence period & active users & \\ \hline \hline
\multirow{2}{*}{PCAC} & $r\left\lfloor\frac{r(q^2+q+1)}{\Delta+1}\right\rfloor$ & $r(q^2+q+1)$ & 
$q+1$ & $\begin{array}{c} q\text{ is a prime power and} \\ r=1 \text{ or } r>q \text{ is a prime} \end{array}$ \\
\cline{2-5}
 & $r\left\lfloor\frac{r(q^2-1)}{\Delta+1}\right\rfloor$ 
& $r(q^2-1)$ & $q$ & $\begin{array}{c} q\text{ is a prime power and} \\ r=1 \text{ or } r\geq q \text{ is a prime} \end{array}$ \\ 
\hline
$\begin{array}{c} \text{GF} \\ \text{RS code} \end{array}$ & $q^m$ & $q^2(\Delta+1)$ & $k$
& $\begin{array}{c} q\text{ is a prime power and} \\ q\geq (k-1)(m-1)+1 \end{array}$ \\ \hline
TDMA & $k$ & $k(\Delta+1)$ & $k$ & \\ 
\hline
\end{tabular}
\caption{Comparison of three approaches \label{tab:comparison}}
\normalsize
\end{table}

We first consider the case that all potential users can be active at the same time;
see Figure~\ref{fig_comparison_1} for examples.
For simple illustration, we fix the number of active users (or potential users) to be $k=p^2+1$ and $\Delta=p^{3/2}$ or $p^2-1$, where $p$ is a prime.
In order to attain $p^2+1$ active users, by Table~\ref{tab:comparison}, the sequence period provided by PCAC approach is at least $p^4+p^2+1$ (i.e., $r=1$ of Case $(a)$), and by GF/RS code approach is at least $(p^2+1)^2(\Delta+1)$ (since the parameter $q\geq p^2+1$ in this case).
Note that the curves of TDMA and PCAC approaches overlap in Figure~\ref{fig_comparison_1} (right) since the original sequence periods provided by them differ by 1 ($p^4+p^2+1$ for PCAC and $p^4+p^2$ for TDMA).

\begin{figure}[h]
\centering
\includegraphics[width=4.7in]{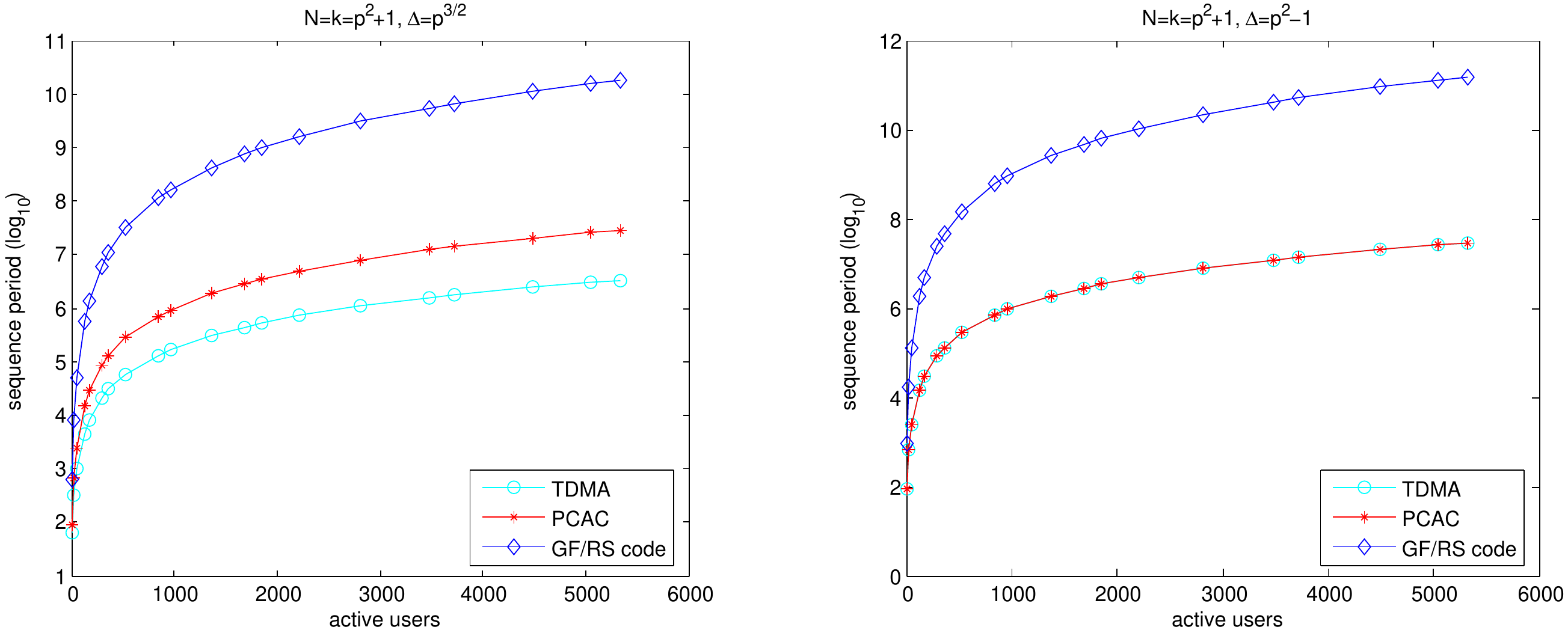}
\caption{$(n,k;(k-1)^{3/4})$-UI and $(n,k;(k-2))$-UI sequence sets for $k=p^2+1$, where $p$ is a prime between $3$ and $73$} \label{fig_comparison_1}
\end{figure}

\noindent
The result reveals that when the number of potential users is almost equal to the maximum number of active users in a system, the TDMA approach has a better performance, where the difference between it with PCAC approach is getting smaller as $\Delta$ approaches $k$.

In practice, however, the number of potential users is much larger than the maximum number of active ones.
Consider the following two cases, shown in Figure~\ref{fig:comparison_2}: 
The number of active users $k$ is set to be a prime $p$, the numbers of potential users is $p^3$, and $\Delta$ is $p-1$ or $p^2-1$.
For PCAC approach, we adopt the Case $(b)$ by letting $r=p$ in the case of $\Delta=p-1$, and $r$ be the smallest prime larger than $p^{3/2}$ in the case of $\Delta=p^2-1$.
By Table~\ref{tab:comparison}, the period of sequences with respect to PCAC (resp. GF/RS code and TDMA) approach is approximately $p^3$ (resp. $4p^3$ and $p^4$) in the first case where $\Delta=k^2$, and approximately $p^{7/2}$ (resp. $4p^4$ and $p^5$) in the second one, where $\Delta=k^3$.
Note that the parameter $m$ in GF/RS code approach is taken to be $3$ to attain the corresponding code size.
One can see that in these two cases, the PCAC approach is much more efficient than the other schemes.

\begin{figure}[h]
\centering
\includegraphics[width=4.7in]{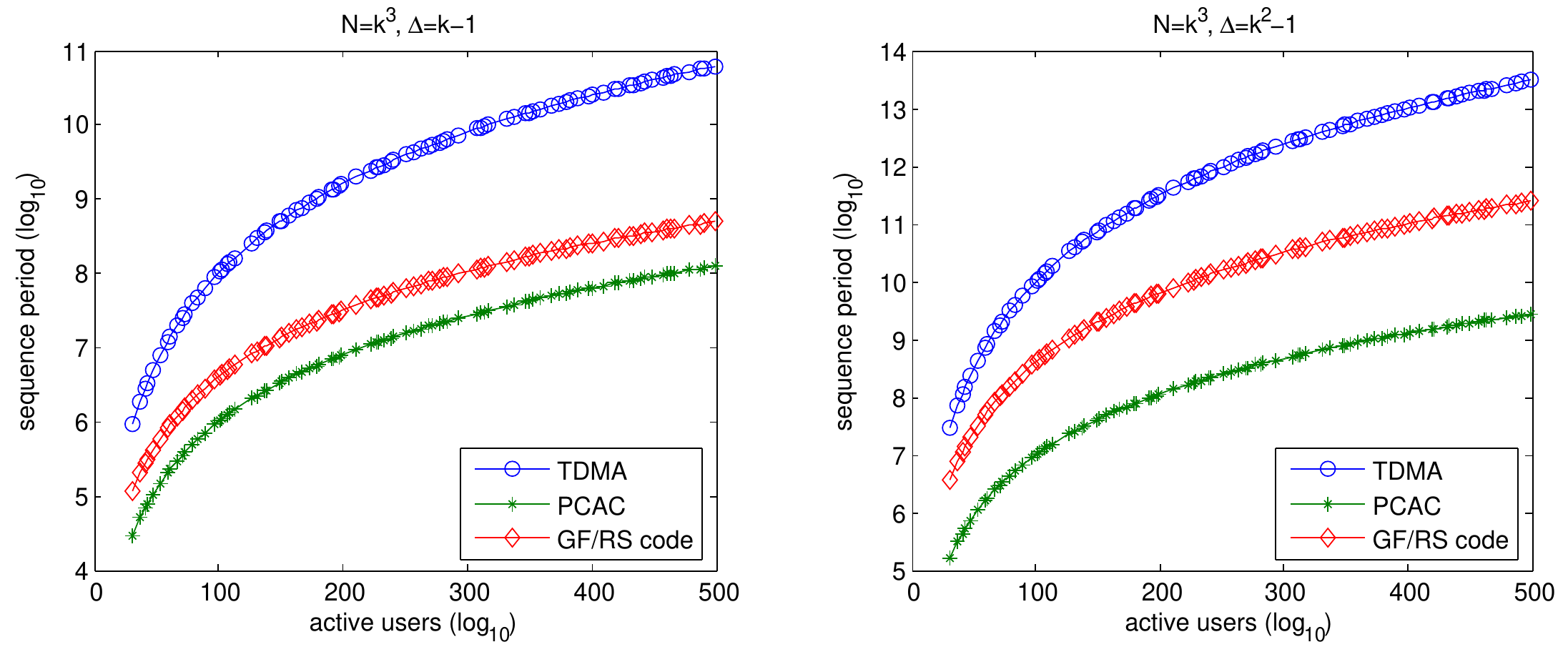}
\caption{$(n,k;k-1)$-UI and $(n,k;k^2-1)$-UI sequence sets with size $k^3$, where $k$ is a prime between 31 and 499} 
\label{fig:comparison_2}
\end{figure}

Roughly speaking, by Table~\ref{tab:comparison}, the PCAC approach provides an $(n,k;\Delta)$-UI sequence set of length $O(k\sqrt{N\Delta})$, while the lengths
of sequences in the TDMA and GF/RS code approaches are respectively $O(N\Delta)$ and $O(\Delta k^2 m^2)$, where $N$ is the code size.
Therefore, the PCAC is more efficient under the condition:
$$k^2 m^4\Delta > N > \frac{k^2}{\Delta}.$$

\section{Partially conflict-avoiding codes of small weight} \label{sec:PCAC23}

In this section, we investigate optimal partially conflict-avoiding codes.
The main technique is to view an optimal $\text{PCAC}_\Delta$ of length $n$ as a maximum packing of $K_n$.
By Lemma~\ref{lem:large-delta}, we only need to consider $\Delta<\lfloor\frac{n}{2}\rfloor$.

\subsection{Weight $k=2$}

Let $i,j$ be the two endpoints of an edge $e$ in $K_n$.
The \emph{difference} of $e$, denoted by $d(e)$, is defined as the smallest nonzero integer $t$ such that
$$i+t\equiv j \text{ (mod $n$) or } j+t\equiv i \text{ (mod $n$).}$$
Note that $1\leq d(e)\leq\frac{n}{2}$ for any edge $e$ in $K_n$.
Note also that in $K_n$ there are exactly $n$ edges of difference $t$ for each $1\leq t < \frac{n}{2}$, and there are exactly $\frac{n}{2}$ edges of difference $\frac{n}{2}$ provided that $n$ is even.
We say an edge $e$ is \emph{exceptional} if $d(e)=\frac{n}{2}$ and is \emph{normal} otherwise.

\begin{lemma}\label{lem:PCAC_k=2}
For $0\leq\Delta<\lfloor\frac{n}2\rfloor$, the maximum size of a $(2,\Delta)$-packing of $K_n$ is $\frac{n-1}{2}\lfloor \frac n{\Delta+1} \rfloor$ if $n$ is odd, and $(\frac{n}{2}-1)\lfloor\frac{n}{\Delta+1}\rfloor + \lfloor\frac{n}{2\Delta+2}\rfloor$ if $n$ is even.
\end{lemma}
\begin{proof}
Assume that $\mathcal{P}$ is a maximum packing.
For each $A\in\mathcal{P}$, the supporting graph $G_A$ is consist of $\Delta+1$ edges with the same difference $d$.
Then the difference $d$ could produce at most $\lfloor\frac{n}{\Delta+1}\rfloor$ supporting graphs if $d<\frac{n}{2}$ or at most $\lfloor\frac{n/2}{\Delta+1}\rfloor$ supporting graphs if $d=\frac{n}{2}$.
Conversely, the construction is straightforward.
Hence the result follows.
\qed
\end{proof}

Combining Lemma~\ref{lem:PCAC_k=2} and Proposition~\ref{pro:PCAC_packing} together with the fact that $M(n,2)=\lfloor\frac{n}{2}\rfloor$, we have:

\begin{theorem} Let $n,\Delta$ be integers with $0\leq\Delta<n$. Then \[ M_\Delta(n,2)=\left\{
\begin{array}{ll}
\frac{n-1}{2}\lfloor \frac n{\Delta+1} \rfloor & \mbox{if $n$ is odd and } \Delta\leq\frac{n-3}2; \\
(\frac{n}{2}-1)\lfloor\frac{n}{\Delta+1}\rfloor + \lfloor\frac{n}{2\Delta+2}\rfloor & \mbox{if $n$ is even and } \Delta\leq\frac{n-2}2; \\
\lfloor \frac n2 \rfloor & \mbox{otherwise.}
\end{array}  \right. \]
\end{theorem}

\subsection{Weight $k=3$}

Let $A$ be a $3$-subset of $\mathbb{Z}_n$ and $\Delta$ be an integer with $0\leq\Delta < \frac n2$.
If two of the three edges in $C_A$ have the same difference, then the number of edges in $G_\Delta(A)$, denoted by $\|G_\Delta(A)\|$, can be determined by the two distinct differences.
For example, let $n=8$.
There are seven edges (four of difference $1$ and three of difference $2$) in $G_2(\{0,1,2\})$, and eight edges (five of difference $2$ and three of difference $4$) in $G_2(\{3,5,7\})$, see Figure~\ref{fig:supporting}.
We characterize this phenomenon below.

\begin{lemma} \label{lem:n=3-size}
Let $A$ be a 3-subset of $\mathbb{Z}_n$ and $\Delta$ be an integer with $0\leq\Delta < \lfloor\frac{n}2\rfloor$.
If there exist two edges in $C_A$ with the same difference $d$ such that $d\neq \frac{n}{3}$, then
$$ \|G_\Delta(A)\|=\left\{
\begin{array}{ll}
2\Delta+2+d & \text{if } d\leq\Delta, \\
3(\Delta+1)& \text{if } d>\Delta,
\end{array}\right. $$
where $\|G_\Delta(A)\|$ is the number of edges in $G_\Delta(A)$.
\end{lemma}
\begin{proof}
Assume $A=\{i,j,k\}$ and $i-j\equiv j-k\equiv d$ (mod $n$).
Let $E_1=\bigcup_{\tau=0}^\Delta \{i+\tau,j+\tau\}$, $E_2=\bigcup_{\tau=0}^\Delta \{j+\tau,k+\tau\}$ and $E_3=\bigcup_{\tau=0}^\Delta \{i+\tau,k+\tau\}$ be the sets of edges in $G_\Delta(A)$.
It is easy to see that $E_1\cap E_2$ is empty if $d>\Delta$ and is equal to $\{i,j\}\cup \cdots \cup\{i+\Delta-d,j+\Delta-d\}$ if $d\leq \Delta$.
That is, there are $\Delta-d+1$ repeated edges if $d\leq \Delta$.
Since $d\neq \frac n3$, $E_1\cap E_3=\emptyset$ and $E_2\cap E_3=\emptyset$.
This completes the proof. \qed
\end{proof}

We note here that if $d=\frac{n}{3}$ in above lemma, then $\|G_\Delta(A)\|=3(\Delta+1)$ if $\frac{n}{3}>\Delta$ and $\|G_\Delta(A)\|=n$ if $\frac{n}{3}\leq\Delta$.
We have the following result.

\begin{lemma} \label{lem:k=3upper}
Given a maximum $(3,\Delta)$-packing $\mathcal{P}$ of $K_n$, where $n$ and $\Delta$ are positive integers with $\Delta<\lfloor\frac{n}2\rfloor$.
Then
\begin{equation}
|\mathcal{P}| < \frac{n(n-1)}{6(\Delta+1)} + \frac{2\ln{2}-1}{3}n + \frac{n}{3(\Delta+1)}.
\end{equation}
\end{lemma}
\begin{proof}
We only consider $3\nmid n$ because the case $3|n$ can be dealt with in the same way.
For $d=1,\ldots,\Delta$, let $T_d\subset\mathcal{P}$ be the collection of 3-subsets $A$ such that in $C_A$, some two edges are of the same difference $d$.
The cardinality of $T_d$ is denoted by $t_d$.
By Lemma~\ref{lem:n=3-size}, each $T_d$ corresponds to $(2\Delta+2+d)t_d$ edges and each of the remaining 3-subsets (not in some $T_d$) corresponds to $3(\Delta+1)$ edges.
Furthermore, every $G_{\Delta}(A)$ for $A\in T_d$ contains exactly $\Delta+d+1$ edges with difference $d$, so $t_d\leq\frac{n}{\Delta+d+1}$.
Then,
\begin{align*}
M &\leq ~t_1+t_2+\cdots+t_\Delta + \frac{{n\choose 2}-((2\Delta+3)t_1+(2\Delta+4)t_2+\cdots+(3\Delta+2)t_\Delta)}{3(\Delta+1)}\\
&= ~\frac{n(n-1)}{6(\Delta+1)} + \frac{\Delta t_1+(\Delta-1)t_2+\cdots +t_\Delta}{3(\Delta+1)}\\
&\leq ~\frac{n(n-1)}{6(\Delta+1)} + \frac{n}{3(\Delta+1)} \sum_{d=1}^\Delta \frac{\Delta+1-d}{\Delta+1+d}.
\end{align*}
Consider the last summation, we have
\begin{align*}
\sum_{d=1}^\Delta \frac{\Delta+1-d}{\Delta+1+d} \leq& ~\int_0^\Delta \left(\frac{\Delta+1-x}{\Delta+1+x}\right) \text{d}x = \int_0^\Delta \left(\frac{2(\Delta+1)}{\Delta+1+x}-1\right) \text{d}x \\
=& ~2(\Delta+1) \ln(\frac{2\Delta+1}{\Delta+1})-\Delta \leq 2(\Delta+1) \ln2 -\Delta,
\end{align*}
and thus the result follows. \qed
\end{proof}

The following result on difference triangle sets can be constructed from \emph{Skolem sequences}~\cite{Skolem_1957} and \emph{hooked Skolem sequences}~\cite{OKeefe_1961}.

\begin{theorem}\cite{OKeefe_1961,Skolem_1957} \label{thm:3-DF}
There exists a $(r,2)$-DTS with scope $3r$ whenever $r\equiv 0,1$ (mod 4), and scope $3r+1$ whenever $r\equiv 2,3$ (mod 4).
\end{theorem}


By Theorem~\ref{thm:DDS_DTS}, there exists an $(n,3,r)$-DDS for all $n\geq 6r+1$ whenever $r\equiv 0,1$ (mod 4), and $n\geq 6r+3$ whenever $r\equiv 2,3$ (mod 4).
Applying Proposition~\ref{pro:DDS_packing} we obtain the following result.

\begin{lemma} \label{lem:k=3lower}
Let $n,\Delta$ be positive integers such that $\Delta<\lfloor\frac{n}2\rfloor$.
There exists a $(3,\Delta)$-packing $\mathcal{P}$ of $K_n$ with 
$$|\mathcal{P}|=\left\lfloor\frac{n-1}6\right\rfloor \left\lfloor\frac{n}{\Delta+1}\right\rfloor.$$
\end{lemma}

The following result can be obtained by Proposition~\ref{pro:PCAC_packing} together with Lemma~\ref{lem:k=3upper} and \ref{lem:k=3lower}.

\begin{theorem}
Let $n,\Delta$ be positive integers such that $\Delta<\lfloor\frac{n}2\rfloor$.
Then 
$$\left\lfloor\frac{n-1}6\right\rfloor \left\lfloor\frac{n}{\Delta+1}\right\rfloor \leq M_\Delta(n,3)\leq \frac{n(n-1)}{6(\Delta+1)} + \frac{2\ln{2}-1}{3}n + \frac{n}{3(\Delta+1)}.$$
\end{theorem}

Table~\ref{tab:bounds_PCAC3} lists some upper and lower bounds of $M_{\Delta}(n,3)$ for $\Delta=\sqrt{n}$, where $n=200t$ for $t=1,2,\ldots,18$.
One can imagine that the larger the value $n$, the smaller the gap between the two bounds with respect to $n$.
Generally speaking, if $\Delta$ is fixed (a constant or a function of $n$), then the code size obtained by disjoint difference sets approximately attains the theoretical upper bound $O(\frac{n^2}{6\Delta})$ as $n\rightarrow \infty$.

\begin{table}[h]
\setlength{\belowcaptionskip}{0pt}
\scriptsize
\begin{tabular}{|c||c|c|c|c|c|c|c|c|c|}
\hline
$n$         & $200$ & $400$ & $600$ & $800$ & $1000$ & $1200$ & $1400$ & $1600$ & $1800$ \\ \hline \hline
Upper bound & $442$ & $1273$  & $2357$ & $3647$ & $5114$ & $6739$ & $8509$ & $10413$ & $12441$ \\ \hline
Lower bound & $429$ & $1254$  & $2277$  & $3591$ & $4980$ & $6567$ & $8388$ & $10374$ & $12259$ \\ \hline \hline
$n$         & $2000$ & $2200$ & $2400$ & $2600$  & $2800$ & $3000$ & $3200$ & $3400$ & $3600$ \\ \hline \hline
Upper bound & $14588$ & $16846$ & $19212$ & $21679$ & $24244$ & $26904$ & $29655$ & $32494$ & $35419$ \\ \hline
Lower bound & $14319$ & $16470$ & $19152$ & $21650$ & $23766$ & $26447$ & $29315$ & $32262$ & $35341$ \\ \hline
\end{tabular}
\caption{Upper and lower bounds on $M_{\sqrt{n}}(n,3)$}
\label{tab:bounds_PCAC3}
\normalsize
\end{table}

\subsection{Weight $k=4,5,6,7$}

Here are some difference family results on $k=4,5,6,7$.

\begin{theorem}\cite{Chen_Zhu_1999,Chen_Zhu_1998,Chen_Zhu_2002} \label{thm:DF_4567}

\begin{enumerate}[(i)]
\item For any prime $p\equiv 1$ (mod $12$) there exists a $(p,4)$-DF.
\item For any prime $p\equiv 1$ (mod $20$) there exists a $(p,5)$-DF.
\item For any prime $p\equiv 1$ (mod $30$) there exists a $(p,6)$-DF with one exception of $p=61$.
\item Let $p\equiv 1$ (mod $42$) be a prime and $p\neq 43,127,211$. Then there exists a $(p,7)$-DF whenever $(-3)^{\frac{p-1}{14}}\neq 1$ in $\mathbb{Z}_p$ or $p<261239791$ or $p>1.236597\times 10^{13}$.
\end{enumerate}
\end{theorem}

Since an $(n,k)$-DF is an $(n,k,\frac{n-1}{k(k-1)})$-DDS, the corresponding $\text{PCAC}_\Delta$s are obtained directly by Proposition~\ref{pro:PCAC_packing} and \ref{pro:DDS_packing}.
In Table~\ref{tab:PCAC_4567} we consider $\Delta=\sqrt{n}$ and list some examples of small $n$ which satisfy conditions in Theorem~\ref{thm:DF_4567}.
We note here that more $\text{PCAC}_\Delta$s, especially for small weights, can be produced by a recursive construction of DTSs with minimum scope \cite{Chu_Colbourn_Golomb_2005}.

\begin{table}[h]
\setlength{\belowcaptionskip}{0pt}
\scriptsize
\begin{tabular}{|c||c|c|c|c|c|c|c|c|c|c|c|c|}
\hline
$n$ & $13$ & $37$ & $61$ & $73$ & $97$ & $109$ & $157$ & $181$ & $193$ & $229$ & $241$ & $277$ \\ \hline
$M_{\sqrt{n}}(n,4)$ & $2$ & $15$ & $30$ & $42$ & $64$ & $81$ & $143$ & $180$ & $192$ & $266$ & $280$ & $345$ \\ \hline \hline
$n$ & $41$ & $61$ & $101$ & $181$ & $241$ & $281$ & $401$ & $421$ & $461$ & $521$ & $541$ & $601$ \\ \hline
$M_{\sqrt{n}}(n,5)$ & $10$ & $18$ & $45$ & $108$ & $168$ & $210$ & $380$ & $399$ & $460$ & $546$ & $594$ & $690$ \\ \hline \hline
$n$ & $31$ & $151$ & $181$ & $211$ & $241$ & $271$ & $331$ & $421$ & $541$ & $571$ & $601$ & $631$ \\ \hline
$M_{\sqrt{n}}(n,6)$ & $4$ & $55$ & $72$ & $91$ & $112$ & $135$ & $187$ & $266$ & $396$ & $418$ & $460$ & $504$ \\ \hline \hline
$n$ & $337$ & $379$ & $421$ & $463$ & $547$ & $631$ & $673$ & $757$ & $883$ & $967$ & $1009$ & $1051$ \\ \hline
$M_{\sqrt{n}}(n,7)$ & $136$ & $162$ & $190$ & $220$ & $286$ & $360$ & $384$ & $468$ & $588$ & $690$ & $720$ & $775$ \\ \hline
\end{tabular}
\caption{Some lower bounds on $M_{\sqrt{n}}(n,k)$ for $k=4,5,6,7$}
\label{tab:PCAC_4567}
\normalsize
\end{table}

\section{Concluding remarks}\label{sec:conclusion}
In this paper we construct an infinite number of new partially UI sequence sets by means of $\text{PCAC}_\Delta$ or disjoint difference sets.
For some particular $n$, we are able to obtain an asymptotically optimal $\text{PCAC}_\Delta$ of length $n$ and weight three.


\end{document}